\documentclass[12pt,3p]{nj_style}

\usepackage{amsmath}
\usepackage{amssymb}
\usepackage{amsthm}
\usepackage{enumerate}
\usepackage[font=small,labelfont=bf,margin=8pt]{caption}
\usepackage[usenames,dvipsnames]{color}
\usepackage{graphicx}
\usepackage[colorlinks=true]{hyperref}
\usepackage{tikz}
\tikzstyle{vertex}=[circle, draw, inner sep=0pt, minimum size=4pt]
\newcommand{\vertex}{\node[vertex]}

\newtheorem{thm}{Theorem} 

\newtheorem{cor}[thm]{Corollary}

\newtheorem{lemma}[thm]{Lemma}

\newcommand{\bra}[1]{\langle #1 |}
\newcommand{\ket}[1]{| #1 \rangle}
\newcommand{\braket}[2]{\langle #1 | #2 \rangle}
\newcommand{\ketbra}[2]{| #1 \rangle\langle #2 |}

\newcommand{\bb}[1]{\mathbb{#1}}

\makeatletter
\renewcommand*\env@matrix[1][c]{\hskip -\arraycolsep
  \let\@ifnextchar\new@ifnextchar
  \array{*\c@MaxMatrixCols #1}}
\makeatother

\begin{document}

\begin{frontmatter}


\title{The Minimum Size of Unextendible Product Bases \\ in the Bipartite Case (and Some Multipartite Cases)}
\author[UG,IQC,AMSS]{Jianxin Chen}
\ead{chenkenshin@gmail.com}

\author[IQC]{Nathaniel Johnston}
\ead{nathaniel.johnston@uwaterloo.ca}

\address[UG]{Department of Mathematics \& Statistics, University of Guelph, Guelph, Ontario N1G~2W1, Canada}
\address[IQC]{Institute for Quantum Computing, University of Waterloo, Waterloo, Ontario N2L~3G1, Canada}
\address[AMSS]{UTS-AMSS Joint Research Laboratory for Quantum Computation and Quantum Information Processing, Academy of Mathematics and Systems Science, Chinese Academy of Sciences, Beijing, China}

\begin{abstract}
	A long-standing open question asks for the minimum number of vectors needed to form an unextendible product basis in a given bipartite or multipartite Hilbert space. A partial solution was found by Alon and Lov\'{a}sz in 2001, but since then only a few other cases have been solved. We solve all remaining bipartite cases, as well as a large family of multipartite cases.
\end{abstract}


\end{frontmatter}

\section{Introduction}

We use ``ket'' notation to represent unit vectors in complex Euclidean space, which represent pure quantum states. We say that a state $\ket{v} \in \bb{C}^{d_1} \otimes \cdots \otimes \bb{C}^{d_p}$ is a \emph{product state} if we can write it in the form
\begin{align*}
	\ket{v} = \ket{v_1} \otimes \cdots \otimes \ket{v_p} \ \ \text{ with } \ \ \ket{v_j} \in \bb{C}^{d_j} \ \forall \, j.
\end{align*}
If no such decomposition of $\ket{v}$ exists, we say that it is \emph{entangled}.

An \emph{unextendible product basis (UPB)}, introduced in \cite{BDMSST99,DMSST03}, is a set of mutually orthogonal product states such that every state in the orthogonal complement of their span is entangled. That is, a set $S \subseteq \bb{C}^{d_1} \otimes \cdots \otimes \bb{C}^{d_p}$ is a UPB if and only if:
\begin{enumerate}[(a)]
	\item $\ket{v}$ is a product state for all $\ket{v} \in S$;
	
	\item $\braket{v}{w} = 0$ for all $\ket{v} \neq \ket{w} \in S$; and
	
	\item for all product states $\ket{z} \notin S$, there exists $\ket{v} \in S$ such that $\braket{v}{z} \neq 0$.
\end{enumerate}

Many applications of UPBs have been found, including the construction of bound entangled states \cite{BDMSST99,LSM11,Sko11b} and indecomposible positive maps \cite{Ter01}. Unextendible product bases have also been shown to give rise to Bell inequalities without a quantum violation \cite{ASHKLA11} and feature the interesting property of nonlocality without entanglement \cite{BDFMRSSW99} -- that is, they can not be perfectly distinguished by local quantum operations and classical communication, even though they contain no entanglement.

One of the oldest questions concerning UPBs asks for their minimum size -- the smallest number of vectors that form a UPB.  A trivial lower bound of $f_N(d_1,\ldots,d_p) := \sum_{j=1}^p(d_j - 1) + 1$ was immediately noted in \cite{BDMSST99}, and it was also shown that this lower bound holds in many cases. In fact, if we remove the orthogonality condition~(b) in the definition of a UPB, then this lower bound is always attained \cite{Bha06}. However, when condition~(b) is present (which is the case we consider), the problem seems to be more difficult. It was shown in \cite{AL01} that there are cases when the lower bound $f_N(d_1,\ldots,d_p)$ is not attained, and furthermore it was determined exactly for which values of $d_1,\ldots,d_p$ this is the case.

The minimal size of UPBs has only been determined in a handful of cases when the trivial lower bound is not attained. Our main contribution to this problem is a solution for a large family of systems, which includes all bipartite (i.e., $p = 2$) systems as a special case. We review known partial answers to this question in Section~\ref{sec:min_size} before presenting our main results. We also briefly introduce our proof technique, which is based on factorizations of the complete graph and tools from algebraic geometry. In Section~\ref{sec:proof} we present the full proof of our main result, and in Section~\ref{sec:tripartite} we prove a related result that answers the question for some additional tripartite (i.e., $p = 3$) systems. As our proofs are largely non-constructive, we discuss how to explicitly construct UPBs of the minimal size in Section~\ref{sec:explicit}. In order to illustrate these procedures, we provide MATLAB code and several examples in small dimensions. We close in Section~\ref{sec:conclusions} with a brief discussion of related questions that remain open.

\section{Minimum Size of Unextendible Product Bases}\label{sec:min_size}

A long-standing open question asks: given the dimensions $d_1, d_2, \ldots, d_p$, what is the minimum possible number of vectors in a UPB $S \subseteq \bb{C}^{d_1} \otimes \cdots \otimes \bb{C}^{d_p}$? We define $f_m(d_1,\ldots,d_p)$ to be this minimum value. Here we briefly list all partial answers to this question that are known to us:
\begin{enumerate}[(1)]
	\item If ${\rm min}(d_1,d_2) = 2$ then $f_m(d_1,d_2) = d_1 d_2$.
	
	\item $f_m(d_1,\ldots,d_p)$ equals the trivial lower bound $f_N(d_1,\ldots,d_p) := \sum_{j=1}^p(d_j - 1) + 1$ if and only if (1) doesn't hold and either $f_N(d_1,\ldots,d_p)$ is even or all $d_j$'s are odd (or both).

	\item If $p \equiv 2 \, (\text{mod } 4)$ then $f_m(2,2,\ldots,2) = p + 2$ \cite{Fen06}.

	\item $f_m(4,4) = 8$ \cite{Ped02}, $f_m(2,2,3) = 6$, $f_m(2,2,5) = 8$, $f_m(2,2,2,2) = 6$, $f_m(2,2,2,4) = 8$, and $f_m(2,2,2,2,5) = 10$ \cite{Fen06}.
\end{enumerate}

Notice in particular that in all known cases except for (1) and (2) above, we have $f_m(d_1,\ldots,d_p) = f_N(d_1,\ldots,d_p) + 1$. Our main result shows that this is a fairly general phenomenon, and there is a rather large class of multipartite systems for which $f_m(d_1,\ldots,d_p) = f_N(d_1,\ldots,d_p) + 1$.
\begin{thm}\label{thm:main}
	Let $2 \leq d_1 \leq d_2 \leq \ldots \leq d_p$ be integers for which neither $(1)$ nor $(2)$ above hold, and $d_p - 1 \geq \sum_{j=1}^{p-1}(d_j - 1) \geq 3$. Then $f_m(d_1,\ldots,d_p) = f_N(d_1,\ldots,d_p) + 1$.
\end{thm}

This result generalizes the known results $f_m(4,4) = f_m(2,2,2,4) = 8$ and $f_m(2,2,2,2,5) = 10$. More importantly, Theorem~\ref{thm:main} solves all remaining bipartite cases. The following corollary follows immediately from Theorem~\ref{thm:main} and the known results (1) and (2).
\begin{cor}
	In the $p = 2$ case, we have the following characterization:
	\begin{enumerate}[(a)]
		\item If ${\rm min}(d_1,d_2) = 2$ then $f_m(d_1,d_2) = d_1 d_2$.
		
		\item If $d_1, d_2 \geq 4$ are even then $f_m(d_1,d_2) = d_1 + d_2$.
		
		\item If neither (a) nor (b) hold then $f_m(d_1,d_2) = d_1 + d_2 - 1$.
	\end{enumerate}
\end{cor}

Note that the only way that the condition $\sum_{j=1}^{p-1}(d_j - 1) \geq 3$ of Theorem~\ref{thm:main} can fail is if either $p = 2$ and $d_1 \leq 3$ (in which case the size of the minimum UPB is given by either condition (1) or (2) above), or if $p = 3$ and $d_1 = d_2 = 2$. Thus, the theorem says nothing about the values of $f_m(2,2,2k+1)$ for $k \geq 3$. The following result solves half of these remaining cases:
\begin{thm}\label{thm:tripartite_filler}
	Let $k$ be a positive integer. Then $f_m(2,2,4k+1) = f_N(2,2,4k+1) + 1$.
\end{thm}

The next two sections are devoted to proving Theorems~\ref{thm:main} and~\ref{thm:tripartite_filler}. Our techniques are very much inspired by the orthogonality graph and graph factorization methods introduced and developed in \cite{BDMSST99,DMSST03,Fen06}.

We illustrate the basic ideas briefly here, before providing the full proofs in the next sections. In all cases, condition (2) above tells us that $f_m(d_1,\ldots,d_p) \geq f_N(d_1,\ldots,d_p) + 1$, so it suffices to construct a UPB of size $f_N(d_1,\ldots,d_p) + 1$. For now, assume for simplicity that we are trying to construct a UPB $S := \{\ket{v_0},\ldots,\ket{v_5},\ket{w_0},\ldots,\ket{w_5}\} \in \bb{C}^6 \otimes \bb{C}^6$. For all $0 \leq j \leq 5$, we can write $\ket{v_j} = \ket{v_j^{(1)}} \otimes \ket{v_j^{(2)}}$ and $\ket{w_j} = \ket{w_j^{(1)}} \otimes \ket{w_j^{(2)}}$, and we know that any two of these product states are orthogonal on at least one of their subsystems (e.g., for all $i\neq j$ we have either $\braket{v_i^{(1)}}{v_j^{(1)}} = 0$ or $\braket{v_i^{(2)}}{v_j^{(2)}} = 0$, or both).

In order to easily keep track of these orthogonality conditions, we use \emph{orthogonality graphs}. All graphs that we consider will be simple graphs. The orthogonality graph of a set of vectors $\{\ket{\phi_0},\ldots,\ket{\phi_{k-1}}\}$ is the graph $(V,E)$ defined by letting $V := \{\phi_0,\ldots,\phi_{k-1}\}$ be a set of $k$ vertices, and $E := \{ (\phi_i,\phi_j) : \braket{\phi_i}{\phi_j} = 0 \}$ be a set of edges so that two vertices are adjacent if and only if the corresponding vectors are orthogonal.

Our goal now is to construct sets of vectors $S^{(1)} := \{\ket{v_0^{(1)}},\ldots,\ket{v_5^{(1)}},\ket{w_0^{(1)}},\ldots,\ket{w_5^{(1)}}\} \in \bb{C}^6$ and $S^{(2)} := \{\ket{v_0^{(2)}},\ldots,\ket{v_5^{(2)}},\ket{w_0^{(2)}},\ldots,\ket{w_5^{(2)}}\} \in \bb{C}^6$ so that their orthogonality graphs $(V,E_1)$ and $(V,E_2)$ satisfy $(V,E_1 \cup E_2) = K_{12}$, the complete graph on $12$ vertices, since this ensures that any two distinct elements of $S$ are orthogonal. For convenience, in the future if we have two graphs $G_1 := (V,E_1)$ and $G_2 := (V,E_2)$, then we define $G_1 \cup G_2 := (V,E_1 \cup E_2)$.

We begin by considering the graph $D_{6,0} := (V,E_1)$, which is constructed by taking two copies of $K_6$ and pairwise joining vertices, as in Figure~\ref{fig:double_complete_graph}. We also consider a second graph, $C_{6,0} := (V,E_2)$, which we define simply to be the complement of $D_{6,0}$ (i.e., an edge is in $E_2$ if and only if it is not in $E_1$), as in Figure~\ref{fig:double_complete_complement}.
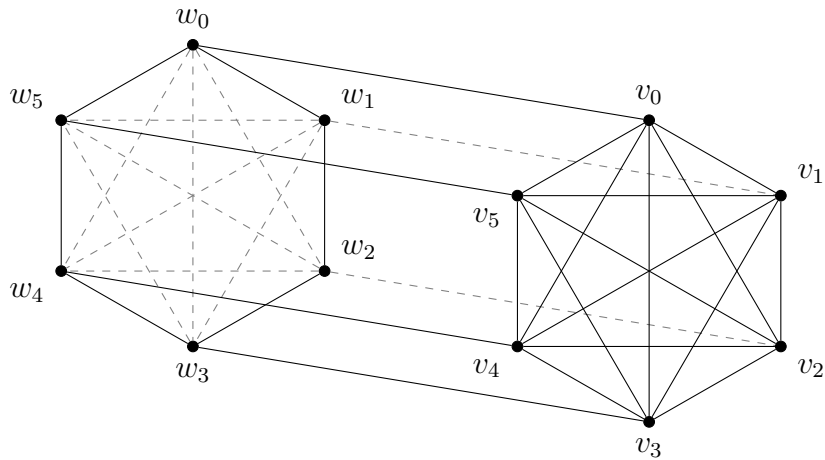
\begin{figure}[htb]
	\centering
	\begin{tikzpicture}[x=2cm, y=2cm, label distance=0cm]
		\vertex[fill] (v0) at (0,1) [label=90:$w_0$]{};
		\vertex[fill] (v1) at (0.866,0.5) [label=15:$w_1$]{};
		\vertex[fill] (v2) at (0.866,-0.5) [label=15:$w_2$]{};
		\vertex[fill] (v3) at (0,-1) [label=270:$w_3$]{};
		\vertex[fill] (v4) at (-0.866,-0.5) [label=195:$w_4$]{};
		\vertex[fill] (v5) at (-0.866,0.5) [label=165:$w_5$]{};
		
		\vertex[fill] (w0) at (3,0.5) [label=90:$v_0$]{};
		\vertex[fill] (w1) at (3.866,0) [label=15:$v_1$]{};
		\vertex[fill] (w2) at (3.866,-1) [label=345:$v_2$]{};
		\vertex[fill] (w3) at (3,-1.5) [label=270:$v_3$]{};
		\vertex[fill] (w4) at (2.134,-1) [label=195:$v_4$]{};
		\vertex[fill] (w5) at (2.134,0) [label=195:$v_5$]{};
			
		\path 
			(v1) edge[dashed,color=gray] (w1)
			(v2) edge[dashed,color=gray] (w2)
			(v0) edge[dashed,color=gray] (v2)
			(v0) edge[dashed,color=gray] (v3)
			(v0) edge[dashed,color=gray] (v4)
			(v1) edge[dashed,color=gray] (v3)
			(v1) edge[dashed,color=gray] (v4)
			(v1) edge[dashed,color=gray] (v5)
			(v2) edge[dashed,color=gray] (v4)
			(v2) edge[dashed,color=gray] (v5)
			(v3) edge[dashed,color=gray] (v5)
			(v0) edge (v1)
			(v0) edge (v5)
			(v1) edge (v2)
			(v2) edge (v3)
			(v3) edge (v4)
			(v4) edge (v5)
			
			(w0) edge (w1)
			(w0) edge (w2)
			(w0) edge (w3)
			(w0) edge (w4)
			(w0) edge (w5)
			(w1) edge (w2)
			(w1) edge (w3)
			(w1) edge (w4)
			(w1) edge (w5)
			(w2) edge (w3)
			(w2) edge (w4)
			(w2) edge (w5)
			(w3) edge (w4)
			(w3) edge (w5)
			(w4) edge (w5)
			
			(v0) edge (w0)
			(v3) edge (w3)
			(v4) edge (w4)
			(v5) edge (w5)
		;
	\end{tikzpicture}
	\caption{The graph $D_{6,0}$ used in the construction of UPBs in which $d_p - 1 = \sum_{j=1}^{p-1}(d_j - 1) = 5$.}\label{fig:double_complete_graph}
\end{figure}

\begin{figure}[htb]
	\centering
	\begin{tikzpicture}[x=2cm, y=2cm, label distance=0cm]
		\vertex[fill] (v0) at (0,1.5) [label=above:$v_0$]{};
		\vertex[fill] (v1) at (1,1.5) [label=above:$v_1$]{};
		\vertex[fill] (v2) at (2,1.5) [label=above:$v_2$]{};
		\vertex[fill] (v3) at (3,1.5) [label=above:$v_3$]{};
		\vertex[fill] (v4) at (4,1.5) [label=above:$v_4$]{};
		\vertex[fill] (v5) at (5,1.5) [label=above:$v_5$]{};
		
		\vertex[fill] (w0) at (0,0) [label=below:$w_0$]{};
		\vertex[fill] (w1) at (1,0) [label=below:$w_1$]{};
		\vertex[fill] (w2) at (2,0) [label=below:$w_2$]{};
		\vertex[fill] (w3) at (3,0) [label=below:$w_3$]{};
		\vertex[fill] (w4) at (4,0) [label=below:$w_4$]{};
		\vertex[fill] (w5) at (5,0) [label=below:$w_5$]{};
		
		\path 
			(v0) edge (w1)
			(v0) edge (w2)
			(v0) edge (w3)
			(v0) edge (w4)
			(v0) edge (w5)
			(v1) edge (w0)
			(v1) edge (w2)
			(v1) edge (w3)
			(v1) edge (w4)
			(v1) edge (w5)
			(v2) edge (w0)
			(v2) edge (w1)
			(v2) edge (w3)
			(v2) edge (w4)
			(v2) edge (w5)
			(v3) edge (w0)
			(v3) edge (w1)
			(v3) edge (w2)
			(v3) edge (w4)
			(v3) edge (w5)
			(v4) edge (w0)
			(v4) edge (w1)
			(v4) edge (w2)
			(v4) edge (w3)
			(v4) edge (w5)
			(v5) edge (w0)
			(v5) edge (w1)
			(v5) edge (w2)
			(v5) edge (w3)
			(v5) edge (w4)
		;
	\end{tikzpicture}
	\caption{The graph $C_{6,0}$, which is the complement of $D_{6,0}$.}\label{fig:double_complete_complement}
\end{figure}
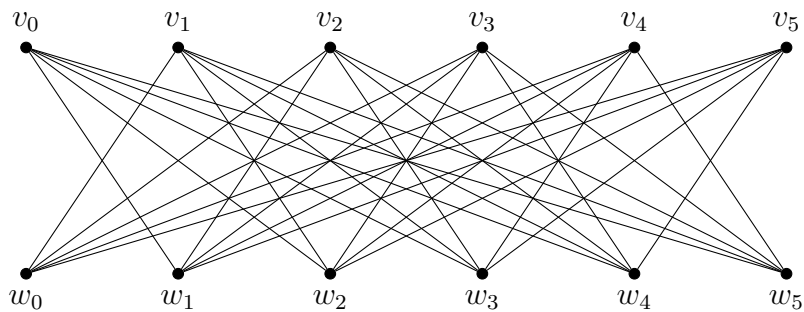

We now must show that we can construct two sets of $12$ vectors in $\bb{C}^6$ that have orthogonality graphs $D_{6,0}$ and $C_{6,0}$. For $C_{6,0}$, this certainly seems like it should be possible, since each vector lives in $\bb{C}^6$ and is orthogonal to exactly $5$ other vectors -- we use tools from algebraic geometry to make this intuition rigorous.

For $D_{6,0}$, the idea is to define $\ket{v_j^{(2)}} := \ket{j}$ to be computational basis states for all $0 \leq j < 6$. Then if we want the orthogonality graph of $S^{(2)}$ to equal $D_{6,0}$, we need to find states $\{\ket{w_j^{(2)}}\}$ that form an orthonormal basis of $\bb{C}^6$ and satisfy $\braket{w_j^{(2)}}{j} = 0$ for all $j$. In other words, we need to find a unitary matrix that has zeroes on its diagonal. We again use methods from algebraic geometry to prove that many such unitary matrices exist, and furthermore that generic sets of vectors with these orthogonality graphs lead to product bases that are indeed unextendible. The ideas presented above generalize without much difficulty to all cases considered by Theorem~\ref{thm:main}.

\section{Proof of Theorem~\ref{thm:main}}\label{sec:proof}

We now present the main body of the proof of Theorem~\ref{thm:main}. We make use of the two technical lemmas~\ref{lem:hollow_unitary} and~\ref{lem:cp_orthogonality} throughout the proof, whose statements and proofs we leave to the end of the section.

Define the quantity $b := \big((d_p - 1) - \sum_{j=1}^{p-1}(d_j - 1)\big)/2 \geq 0$, which we can assume is an integer since condition~(2) of Section~\ref{sec:min_size} solves the case when it is a half-integer. Define $D_{d_p,b}$ to be the graph on $2(d_p - b) = f_N(d_1,\ldots,d_p) + 1$ vertices $V := \{v_0,\ldots,v_{d_p-b-1},w_0,\ldots,w_{d_p-b-1}\}$ such that $v_i$ is adjacent to $v_j$ and $w_i$ is adjacent to $w_j$ for all $i \neq j$, and $v_i$ is adjacent to $w_j$ if and only if $j - i \equiv s \, (\text{mod } (d_p - b))$ for some $0 \leq s \leq b$. In the $d_p = 6, b = 0$ case, this graph is depicted in Figure~\ref{fig:double_complete_graph}. The graph $D_{7,1}$ looks the same, except it has $6$ additional edges: $(v_0,w_1), (v_1,w_2), \ldots, (v_5,w_0)$. 

Our first goal is to construct a set of $2(d_p - b)$ vectors
\begin{align*}
	S^{(p)} := \Big\{\ket{v_0^{(p)}},\ldots,\ket{v_{d_p-b-1}^{(p)}},\ket{w_0^{(p)}},\ldots,\ket{w_{d_p-b-1}^{(p)}}\Big\} \in \bb{C}^{d_p}
\end{align*}
such that the orthogonality graph of $S^{(p)}$ is $D_{d_p,b}$. To this end, let $\ket{v_j^{(p)}} := \ket{j}$ be standard basis states for all $0 \leq j < d_p - b$. We then need to find states $\{\ket{w_j^{(p)}}\}_{j=0}^{d_p-b-1}$ that form an orthonormal set and have the additional property that $\braket{(j + \ell) \, (\text{mod } d)}{w_j^{(p)}} = 0$ for all $0 \leq \ell \leq b$. We choose these states to be the (normalized) columns of a matrix $U$ described by Lemma~\ref{lem:hollow_unitary}, which is proved at the end of this section. Importantly, note that any $d_p+1$ distinct elements of $S^{(p)}$ span all of $\bb{C}^{d_p}$. To see this, suppose we choose $r+1$ $\ket{w_j^{(p)}}$'s and $d_p - r$ $\ket{v_j^{(p)}}$'s from $S^{(p)}$ (for some $b \leq r < d_p - b$). It then suffices to show that every $r \times (r+1)$ submatrix of $U$ has full rank $r$. When $b+2 \leq r \leq d_p-b-2$, this claim follows immediately from condition $(iii)$ of Lemma~\ref{lem:hollow_unitary}. Similarly, condition $(iii)$ says that every $(b+2)\times(b+2)$ submatrix of $U$ has full rank, so every $(b+1)\times(b+2)$ submatrix of $U$ must have rank $b+1$, which proves the $r = b+1$ case. The $r = d_p-b-1$ case is similar. The last remaining case is $r = b$, in which case it is not true that every $r \times (r+1)$ submatrix of $U$ has full rank $r$. Nonetheless, because all $d_p - b$ $\ket{v_j^{(p)}}$'s have been chosen, it suffices to just show that every $b \times b$ submatrix contained entirely within the bottom $b$ rows of $U$ is nonsingular, which is guaranteed by $(iv)$ of Lemma~\ref{lem:hollow_unitary}. This completes the proof that any $d_p+1$ distinct elements of $S^{(p)}$ span all of $\bb{C}^{d_p}$.

Now we consider the complement of the graph $D_{d_p,b}$, which we denote $C_{d_p,b}$. This is a graph on the same set of $2(d_p - b)$ vertices $V$, but $(v_i,v_j)$ and $(w_i,w_j)$ are never edges in $C_{d_p,b}$, and $(v_i,w_j)$ is an edge in $C_{d_p,b}$ if and only if $j - i \equiv \ell \, (\text{mod } (d_p - b))$ for some $b < \ell < d_p - b$. For example, in the $d_p = 6, b = 0$ case, this graph is depicted in Figure~\ref{fig:double_complete_complement}. The graph $C_{7,1}$ looks the same, but without the edges $(v_0,w_1), (v_1,w_2), \ldots, (v_5,w_0)$.

We now split $C_{d_p,b}$ into $p - 1$ spanning subgraphs -- one for each of the first $p - 1$ systems. Let $C_{d_p,b,r,s}$ ($r,s \geq 1$ and $s+r \leq d_p - b$) be the subgraph of $C_{d_p,b}$ that contains the edge $(v_i, w_j)$ if and only if $j - i \equiv \ell \, (\text{mod } (d_p - b))$ for some $s \leq \ell < s + r$. For example, the graphs $C_{6,0,2,1}$ and $C_{6,0,3,3}$ are depicted in Figure~\ref{fig:double_complete_complement_decomp}. To help keep in mind which indices represent what, we can think of $r$ as the regularity of the graph (i.e., $C_{d_p,b,r,s}$ is $r$-regular -- every vertex has degree $r$), and $s$ as a shift (since, for example, the smallest index $j$ such that $(v_0,w_j)$ is an edge in $C_{d_p,b,r,s}$ is $j = s$). It is straightforward to see that we have $C_{d_p,b,r_1,s_1} \cup \cdots \cup C_{d_p,b,r_k,s_k} = C_{d_p,b}$ whenever $\sum_{j=1}^{k} r_j = d_p - 2b - 1$ and $s_1 = 1$, $s_j = s_{j-1} + r_{j-1}$ for $1 < j \leq k$. We choose $k = p - 1$ and $r_j = d_j - 1$ for all $1 \leq j \leq p-1$.
\begin{figure}[htb]
	\centering
	\begin{tikzpicture}[x=1cm, y=1cm, label distance=0cm]
		\vertex[fill] (v01) at (0,1.5) [label=above:$v_0$]{};
		\vertex[fill] (v11) at (1,1.5) [label=above:$v_1$]{};
		\vertex[fill] (v21) at (2,1.5) [label=above:$v_2$]{};
		\vertex[fill] (v31) at (3,1.5) [label=above:$v_3$]{};
		\vertex[fill] (v41) at (4,1.5) [label=above:$v_4$]{};
		\vertex[fill] (v51) at (5,1.5) [label=above:$v_5$]{};
		
		\vertex[fill] (w01) at (0,0) [label=below:$w_0$]{};
		\vertex[fill] (w11) at (1,0) [label=below:$w_1$]{};
		\vertex[fill] (w21) at (2,0) [label=below:$w_2$]{};
		\vertex[fill] (w31) at (3,0) [label=below:$w_3$]{};
		\vertex[fill] (w41) at (4,0) [label=below:$w_4$]{};
		\vertex[fill] (w51) at (5,0) [label=below:$w_5$]{};

		\vertex[fill] (v02) at (7,1.5) [label=above:$v_0$]{};
		\vertex[fill] (v12) at (8,1.5) [label=above:$v_1$]{};
		\vertex[fill] (v22) at (9,1.5) [label=above:$v_2$]{};
		\vertex[fill] (v32) at (10,1.5) [label=above:$v_3$]{};
		\vertex[fill] (v42) at (11,1.5) [label=above:$v_4$]{};
		\vertex[fill] (v52) at (12,1.5) [label=above:$v_5$]{};
		
		\vertex[fill] (w02) at (7,0) [label=below:$w_0$]{};
		\vertex[fill] (w12) at (8,0) [label=below:$w_1$]{};
		\vertex[fill] (w22) at (9,0) [label=below:$w_2$]{};
		\vertex[fill] (w32) at (10,0) [label=below:$w_3$]{};
		\vertex[fill] (w42) at (11,0) [label=below:$w_4$]{};
		\vertex[fill] (w52) at (12,0) [label=below:$w_5$]{};
		
		\path 
			(v01) edge (w11)
			(v01) edge (w21)
			(v11) edge (w21)
			(v11) edge (w31)
			(v21) edge (w31)
			(v21) edge (w41)
			(v31) edge (w41)
			(v31) edge (w51)
			(v41) edge (w51)
			(v41) edge (w01)
			(v51) edge (w01)
			(v51) edge (w11)

			(v02) edge (w32)
			(v02) edge (w42)
			(v02) edge (w52)
			(v12) edge (w42)
			(v12) edge (w52)
			(v12) edge (w02)
			(v22) edge (w02)
			(v22) edge (w12)
			(v22) edge (w52)
			(v32) edge (w02)
			(v32) edge (w12)
			(v32) edge (w22)
			(v42) edge (w12)
			(v42) edge (w22)
			(v42) edge (w32)
			(v52) edge (w22)
			(v52) edge (w32)
			(v52) edge (w42)
			;
	\end{tikzpicture}
	\caption{The graph $C_{6,0}$, decomposed into the union of $C_{6,0,2,1}$ and $C_{6,0,3,3}$.}\label{fig:double_complete_complement_decomp}
\end{figure}
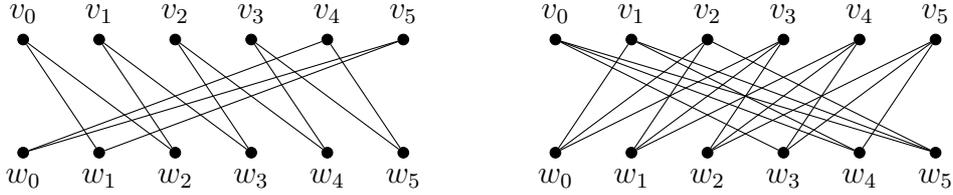

Our goal now, for each $1 \leq j < p$, is to find a set of $2(d_p - b)$ vectors $S^{(j)}$ whose orthogonality graph is $C_{d_p,b,d_j-1,s_j}$. We choose $S^{(j)}$ to be the (normalized) columns of the matrix $V$ described by Lemma~\ref{lem:cp_orthogonality} (with $q = d_p - b$, $r = d_j - 1$, and $s = s_j$). Not only do these vectors have the desired orthogonality graph, but property~(ii) of the lemma ensures that any $d_j$ distinct elements of $S^{(j)}$ ($1 \leq j < p$) span all of $\bb{C}^{d_j}$.

Since we have $C_{d_p,b,d_1-1,s_1} \cup \cdots \cup C_{d_p,b,d_{p-1}-1,s_{p-1}} \cup D_{d_p,b} = C_{d_p,b} \cup D_{d_p,b} = K_{2(d_p-b)}$, the complete graph on $2(d_p-b)$ vertices, it follows that the set of product states
\begin{align*}
	S := \left\{ \bigotimes_{i=1}^p \ket{v_0^{(i)}}, \ldots, \bigotimes_{i=1}^p \ket{v_{d_p-b-1}^{(i)}}, \bigotimes_{i=1}^p \ket{w_0^{(i)}}, \ldots, \bigotimes_{i=1}^p \ket{w_{d_p-b-1}^{(i)}} \right\}
\end{align*}
are mutually orthonormal. To see that this set is unextendible, recall that any $d_p + 1$ of the vectors in $S^{(p)}$ span all of $\bb{C}^{d_p}$ and any $d_j$ of the vectors in $S^{(j)}$ span all of $\bb{C}^{d_j}$ for $1 \leq j < p$. Thus, any product state not in $S$ can be orthogonal to at most $d_p$ elements of $S$ on the $p$-th subsystem and at most $d_j - 1$ elements of $S$ on the $j$-th subsystem for $1 \leq j < p$, and thus it can be orthogonal to a total of at most $\sum_{j=1}^{p-1}(d_j - 1) + d_p = f_N(d_1,\ldots,d_p)$ elements of $S$. However, there are $f_N(d_1,\ldots,d_p) + 1$ vectors in $S$, so no product state is orthogonal to them all, which completes the proof.

We now present and prove the two technical lemmas that we made use of earlier in this section.
\begin{lemma}\label{lem:hollow_unitary}
	For all integers $b \geq 0$ and $d \geq 4+2b$, there exists a matrix $U \in M_{d,d-b}$ such that the following four conditions hold:
	\begin{enumerate}[(i)]
		\item the columns of $U$ form an orthonormal set;
		\item $\bra{(j + \ell) \, (\text{mod } (d-b))}U\ket{j} = 0$ for all $0 \leq j < d-b$ and $0 \leq \ell \leq b$;
		\item all $r \times r$ submatrices of $U$ are nonsingular whenever $r \in \{b+2,b+3,\ldots,d-b\} \setminus \{d-b-1\}$; and
		\item all $b \times b$ submatrices of $U$ that are contained entirely within the bottom $b$ rows of $U$ are nonsingular.
	\end{enumerate}
\end{lemma}
\begin{proof}
For any $U \in M_{d,d-b}$, we write
\begin{eqnarray*}
U=\left(
\begin{array}{cccc}
x_{1,1}+iy_{1,1}  & x_{1,2}+iy_{1,2} & \cdots & x_{1,d-b}+iy_{1,d-b}  \\
x_{2,1}+iy_{2,1}  & x_{2,2}+iy_{2,2} & \cdots & x_{2,d-b}+iy_{2,d-b}  \\
\vdots & \vdots & \ddots & \vdots\\
x_{d,1}+iy_{d,1}  & x_{d,2}+iy_{d,2} & \cdots & x_{d,d-b}+iy_{d,d-b}  
\end{array}
\right)
\end{eqnarray*}
where $x_{p,q}, y_{p,q}\in \mathbb{R}$ for all $1\leq p\leq d, 1\leq q\leq d-b$.

Then conditions $(i)$ and $(ii)$ can be rewritten as follows:
\begin{enumerate}[$(i^{\prime})$]
\item $x_{(j + \ell)(\text{mod } (d-b)),j}=y_{(j + \ell)(\text{mod } (d-b)),j}=0$ for all $1 \leq j \leq d-b$ and $0 \leq \ell \leq b$, 
\item $\sum\limits_{k=1}^{d} (x_{kp}-iy_{kp})(x_{kq}+iy_{kq})=0 \textrm{\ }\forall \, 1\leq p<q\leq d-b$,
\end{enumerate}

We are mainly interested in the real variety $Z\subseteq {\mathbb{R}}^{2d(d-b)}$ characterized by conditions $(i^{\prime})$ and $(ii^{\prime})$. We also look into the real variety $Z_{i_1i_2\ldots i_r,j_1j_2\ldots j_r}\subseteq {\mathbb{R}}^{2d(d-b)}$ characterized by polynomials of $x_{p,q}$'s and $y_{p,q}$'s arising from the constraint $\det(U_{i_1i_2\ldots i_r,j_1j_2\ldots j_r})=0$, where $U_{i_1i_2\cdots i_r,j_1j_2\cdots j_r}$ is the $r \times r$ submatrix of $U$ formed by rows $i_1,i_2,\ldots, i_r$ and columns $j_1,j_2,\ldots,j_r$. 

Our aim is to show that
\begin{align*}
	Z\setminus\Big( \big(\bigcup\limits_{r \in \{b+2,b+3,\ldots,d-b\} \setminus \{d-b-1\} \atop {1\leq i_1<\cdots <i_r\leq d\atop 1\leq j_1<\cdots<j_r\leq d-b}}  Z_{i_1i_2\ldots i_r,j_1j_2\ldots j_r}\big)\cup \big(\bigcup\limits_{ {d-b+1\leq i_1<\cdots <i_{b}\leq d\atop 1\leq j_1<\cdots<j_{b}\leq d-b}} Z_{i_1i_2\ldots i_{b}, j_1j_2\ldots j_{b}}\big)\Big) \neq \emptyset.
\end{align*}

Note that if $Z_{i_1i_2\ldots i_r,j_1j_2\ldots j_r}\cap Z$ is a proper subset of $Z$, then $Z_{i_1i_2\ldots i_r,j_1j_2\ldots j_r}\cap Z$ has zero measure in $Z$. In order to show that $Z_{i_1i_2\ldots i_r,j_1j_2\ldots j_r}\cap Z$ is a proper subset of $Z$, we first need some simple definitions. A \emph{permutation matrix} is a square binary matrix that has exactly one entry $1$ in each row and each column and $0$'s elsewhere. Similarly, a \emph{column-permutation matrix} is a (not necessarily square) binary matrix that has exactly one entry $1$ in each column, at most one entry $1$ in each row, and $0$'s elsewhere. Clearly, every permutation matrix $W\in M_{d,d}$ has full rank $d$ and every column-permutation matrix $W\in M_{d,d-b}$ has rank $d-b$.

Let $W\in M_{d,d-b}$ be the matrix with a $1$ in the $({(j+b+1)(\text{mod }(d-b)),j})$-entry for all $1 \leq j \leq d-b$, and $0$'s elsewhere. The top $(d-b) \times (d-b)$ submatrix of $W$ is a permutation matrix, so it is nonsingular. Therefore the corresponding point in $\mathbb{R}^{2d(d-b)}$ that represents $W$ lies in $Z\setminus Z_{12\ldots (d-b),12\ldots (d-b)}$, which implies that $Z_{12\ldots (d-b),12\ldots (d-b)}\cap Z$ is a proper subset of $Z$. Hence $Z_{12\ldots (d-b),12\ldots (d-b)}\cap Z$ has measure zero in $Z$.

Similarly, for any $1\leq i_1<\cdots <i_{d-b}\leq d$, we can always move some appropriately chosen nonzero rows of the above $W$ to its lower $b \times (d-b)$ submatrix such that the corresponding point in $\mathbb{R}^{2d(d-b)}$ lies in some $Z\setminus Z_{i_1i_2\ldots i_{(d-b)},12\ldots (d-b)}$. Hence,  $Z\cap Z_{i_1i_2\ldots i_{(d-b)},12\ldots (d-b)}$ has measure zero in $Z$.

Next, for any fixed $1\leq i_1<\cdots <i_r\leq d$ and $1\leq j_1<\cdots<j_r\leq d-b$, where $b+2\leq r\leq d-b-2$, if there is some column-permutation matrix $W\in M_{d,d-b}$ satisfying condition $(ii)$ such that $W_{i_1i_2\ldots i_r,j_1j_2\ldots j_r}$ is a permutation matrix, then the point in $\mathbb{R}^{2d(d-b)}$ that represents $W$ lies in $Z\setminus (Z_{12\ldots (d-b), 12\ldots (d-b)}\cup Z_{i_1i_2\ldots i_r,j_1j_2\ldots j_r})$. It follows that if such a matrix $W$ exists then $Z\cap (Z_{12\ldots (d-b), 12\ldots (d-b)}\cup Z_{i_1i_2\ldots i_r,j_1j_2\ldots j_r})$ also has measure zero in $Z$.

We can construct $W\in M_{d,d-b}$ to satisfy these requirements as follows. We choose the submatrix formed by rows $i_1,i_2,\ldots, i_r$ and columns  $j_1,j_2,\ldots,j_r$ to be a permutation submatrix and the submatrix formed by rows $\{1,2,\ldots,d\}\setminus \{i_1,i_2,\ldots ,i_r\}$ and columns $\{1,2,\ldots, d-b\}\setminus \{j_1,j_2,,\ldots,j_r\}$ to be a column-permutation submatrix. All other entries are chosen to be zero. 

To see that the above two submatrices of $W$ can be chosen as desired, note that condition $(ii)$ forces $(b+1)(d-b)$ entries of $W\in M_{d,d-b}$ to be zero. For any given $r \times r$ submatrix $W_{i_1i_2\ldots i_r, j_1j_2\ldots j_r}$ formed by rows $i_1,i_2,\ldots, i_r$ and columns $j_1,j_2,\ldots,j_r$, where $b+2\leq r\leq d-b-2$,  there are at most $(b+1)$ preset zero entries in each column $j_t$ $(1\leq t\leq r)$. So we can first fix a nonpreset entry $(i_1,j_1)$ to be $1$ and then seek an $(r-1) \times (r-1)$ permutation submatrix $W_{i_2\ldots i_r, j_2\ldots j_r}$. This observation implies that we only need to prove the existence of $(b+2) \times (b+2)$ permutation submatrix formed by any $(b+2)$ rows and $(b+2)$ columns that satisfy condition $(ii)$. We also require the existence of $(b+2) \times 2$ column-permutation submatrix formed by any $(b+2)$ rows and any $2$ columns that satisfy condition $(ii)$, but this claim is obvious. Hence, the only thing left is to prove is that, given fixed $1\leq i_1<\cdots <i_{b+2}\leq d$ and $1\leq j_1<\cdots<j_{b+2}\leq d-b$, there exists $W \in M_{d,d-b}$ satisfying condition $(ii)$ such that $W_{i_1i_2\ldots i_{b+2},j_1j_2\ldots j_{b+2}}$ is a permutation matrix.


To see that such a matrix $W$ exists, note that there is no $x \times y$ submatrix of $W$ with $x+y = b+3$ that is forced entirely to equal zero by condition $(ii)$. Thus the same is true of all $(b+2)\times(b+2)$ submatrices of $W$. Thus there is no $(b+2)\times(b+2)$ submatrix of $W$ that satisfies condition $(ii)$ of \cite[Theorem~3.1]{HS93} with $m = n = b+2$, $r = b+1$. By then using condition $(iii)$ of \cite[Theorem~3.1]{HS93} and Lemma~3.7 of the same paper, we see that every $(b+2)\times(b+2)$ submatrix of $W$ can be made nonsingular by a suitable choice of the entries that are unaffected by condition~$(ii)$. If we write the determinant as the sum of products of entries of the $(b+2) \times (b+2)$ submatrix, then at least one of the terms in the sum can be made nonzero -- the permutation matrix we want is the one that corresponds to this term in the sum.

We have shown that $Z \cap Z_{i_1i_2\ldots i_r,j_1j_2\ldots j_r} $ has measure zero in $Z$ for any $r \in \{b+2,b+3,\ldots,d-b\}\setminus \{d-b-1\}$, $1\leq i_1<\cdots <i_r\leq d$ and $1\leq j_1<\cdots<j_r\leq d-b$.

For any ${d-b+1\leq i_1<\cdots <i_{b}\leq d}$ and ${1\leq j_1<\cdots<j_{b}\leq d-b}$, we choose suitable $W\in M_{d,d-b}$ such that its $b \times b$ submatrix formed by rows $i_1,i_2,\ldots, i_b$ and columns $j_1,j_2,\ldots,j_b$ is a permutation matrix and all other entries are fixed to be zero. Its corresponding point in $\mathbb{R}^{2d(d-b)}$ lies in $Z\setminus Z_{i_1i_2\ldots i_{b},j_1j_2\ldots j_b}$. Hence,  $Z\cap Z_{i_1i_2\ldots i_b,j_1j_2\ldots j_b}$ has measure zero in $Z$.

This implies that
\begin{align*}
	Z\bigcap \Big( \big(\bigcup\limits_{r \in \{b+2,b+3,\ldots,d-b\} \setminus \{d-b-1\} \atop {1\leq i_1<\cdots <i_r\leq d\atop 1\leq j_1<\cdots<j_r\leq d-b}}  Z_{i_1i_2\ldots i_r,j_1j_2\ldots j_r}\big)\cup \big(\bigcup\limits_{ {d-b+1\leq i_1<\cdots <i_{b}\leq d\atop 1\leq j_1<\cdots<j_{b}\leq d-b}} Z_{i_1i_2\ldots i_{b}, j_1j_2\ldots j_{b}}\big)\Big)
\end{align*}
also has measure zero in $Z$ since the union of a finite number of measure zero subsets is again a measure zero subset. In particular, it follows that
\begin{align*}
	Z\setminus\Big( \big(\bigcup\limits_{r \in \{b+2,b+3,\ldots,d-b\} \setminus \{d-b-1\} \atop {1\leq i_1<\cdots <i_r\leq d\atop 1\leq j_1<\cdots<j_r\leq d-b}}  Z_{i_1i_2\ldots i_r,j_1j_2\ldots j_r}\big)\cup \big(\bigcup\limits_{ {d-b+1\leq i_1<\cdots <i_{b}\leq d\atop 1\leq j_1<\cdots<j_{b}\leq d-b}} Z_{i_1i_2\ldots i_{b}, j_1j_2\ldots j_{b}}\big)\Big) \neq \emptyset.
\end{align*}

This implies the existence of matrix $U \in M_{d,d-b}$ with properties $(i), (ii), (iii)$ and $(iv)$. Indeed, a generic matrix $U \in M_{d,d-b}$ satisfying $(i)$ and $(ii)$ will also satisfy conditions $(iii)$ and $(iv)$. 
\end{proof}

\begin{lemma}\label{lem:cp_orthogonality}
	Let $q,r,s$ be positive integers satisfying $q \geq r+s$. Then there exists a matrix $V \in M_{r+1,2q}$ such that:
	\begin{enumerate}[(i)]
		\item for all $0 \leq j < q$ and $s \leq \ell < s+r$, $V\ket{j}$ is orthogonal to $V\ket{(j+\ell)(\text{mod } q) + q}$; and
		\item every $(r+1)\times(r+1)$ submatrix of $V$ is nonsingular.
	\end{enumerate}
\end{lemma}
\begin{proof}
For any $V \in M_{r+1,2q}$, we write
\begin{eqnarray*}
V=\left(
\begin{array}{cccc}
x_{1,1}+iy_{1,1}  & x_{1,2}+iy_{1,2}  & \cdots & x_{1,2q}+iy_{1,2q}  \\
x_{2,1}+iy_{2,1}  & x_{2,2}+iy_{2,2}  & \cdots & x_{2,2q}+iy_{2,2q}  \\
\vdots & \vdots & \ddots & \vdots\\
x_{r+1,1}+iy_{r+1,1}  & x_{r+1,2}+iy_{r+1,2}  & \cdots & x_{r+1,2q}+iy_{r+1,2q}  
\end{array}
\right)
\end{eqnarray*}
where $x_{i,j}, y_{i,j}\in \mathbb{R}$ for all $1\leq i\leq r+1, 1\leq j\leq 2q$.

Condition $(i)$ implies that
\begin{eqnarray}
\bra{j} V^{\dagger} V \ket{(j+\ell)(\text{mod } q) + q}=0
\end{eqnarray}
for all $0 \leq j < q$ and $s \leq \ell < s+r$. Equivalently, 
\begin{eqnarray}
\sum\limits_{k=1}^{r+1} (x_{k,j+1}-i y_{k,j+1})(x_{k,(j+\ell)(\text{mod } q)+q+1}+i y_{k,(j+\ell)(\text{mod } q)+q+1})=0
\end{eqnarray}
for all $0 \leq j < q$ and $s \leq \ell < s+r$.

So every matrix $V \in M_{r+1,2q}$ satisfying condition $(i)$ corresponds to a solution of the above polynomial system. Let $Z\subseteq \mathbb{R}^{4(r+1)q}$ be the real variety that is characterized by this polynomial system. We also look into the real variety $Z_{j_1j_2\cdots j_{r+1}}\subseteq {\mathbb{R}}^{4(r+1)q}$ characterized by polynomials of the $x_{i,j}$'s and $y_{i,j}$'s arising from the condition $\det(V_{j_1j_2\ldots j_{r+1}})=0$, where $V_{j_1j_2\ldots j_{r+1}}$ is the $(r+1) \times (r+1)$ matrix formed by columns $j_1,j_2,\ldots,j_{r+1}$ of $V$. 

To show that
\begin{align*}
	Z\setminus\Big( \bigcup\limits_{1\leq j_1<\cdots<j_{r+1}\leq 2q}  Z_{j_1j_2\ldots j_{r+1}}\Big)\neq \emptyset,
\end{align*}
it suffices to prove that $Z \cap Z_{j_1j_2\ldots j_{r+1}}$ is a proper subset of $Z$ for any $1\leq j_1<\cdots<j_{r+1}\leq 2q$.  To see why this is sufficient, note that if this is the case then $Z \cap Z_{j_1j_2\cdots j_{r+1}}$ has zero measure in $Z$ for all $1\leq j_1<\cdots<j_{r+1}\leq 2q$. A finite union of measure zero subsets is again a measure zero subset, which implies  $Z\cap \big(\bigcup\limits_{1\leq j_1<\cdots<j_{r+1}\leq 2q}  Z_{j_1j_2\ldots j_{r+1}}\big)$ has measure zero in $Z$. The lemma then follows straightforwardly, since a generic matrix in $M_{r+1,2q}$ that satisfies condition $(i)$ will also satisfy condition $(ii)$.

To complete the proof, we thus just need to show that, for any fixed $1\leq j_1<\cdots<j_{r+1}\leq 2q$, there exists some $V^\prime \in M_{r+1,2q}$ such that, for all $0 \leq j < q$ and $s \leq \ell < s+r$, $V^\prime\ket{j}$ is orthogonal to $V^\prime\ket{(j+\ell)(\text{mod } q) + q}$ and the $(r+1)\times(r+1)$ submatrix formed by columns $j_1,j_2,\ldots,j_{r+1}$ of $V^\prime$ is nonsingular.

To this end, we first choose the $j_1,j_2,\ldots,j_{r+1}$-th column vectors of $V^\prime$ to be the set of computational basis states in $(r+1)$-dimensional space. Thus the submatrix formed by these columns is nonsingular. We then fill the other columns with suitable computational basis states one by one. To fill the $k$-th column vector, note that it is orthogonal to $r$ other column vectors by condition~$(i)$. If $r^{\prime}$ of these $r$ column vectors have already been determined, then we can choose the $k$-th column vector to be any computational vector that is orthogonal to the $r^{\prime} \leq r$ already-determined column vectors. By repeating this procedure, we eventually fill the entire matrix $V^\prime$ so that it satisfies condition $(i)$ and satisfies condition $(ii)$ for the particular submatrix formed by columns $j_1,j_2,\ldots,j_{r+1}$.
\end{proof}

\section{Proof of Theorem~\ref{thm:tripartite_filler}}\label{sec:tripartite}

Our proof of Theorem~\ref{thm:tripartite_filler} is similar in style to that of Theorem~\ref{thm:main}. As it is already known that $f_m(2,2,4k+1) \geq f_N(2,2,4k+1) + 1$, it suffices to find a UPB of size $f_N(2,2,4k+1) + 1$. To this end, we begin by presenting a graph that leads to a product basis of the desired size, and then we use algebraic geometry techniques to show that many of these product bases are unextendible.

We begin by defining the graph $Y_{4k+4}$ on $f_N(2,2,4k+1) + 1 = 4k + 4$ vertices $V := \{v_0,\ldots,v_{2k+1},w_0,\ldots,w_{2k+1}\}$ such that $v_i$ ($w_i$) is adjacent to $v_j$ ($w_j$) if and only if $j - i \in \{1,-1\} \, (\text{mod } (2k+1))$, and $v_i$ is adjacent to $w_j$ if and only if $i = j$ (this graph is sometimes called the \emph{prism graph} on $4k+4$ vertices). In the $k = 2$ case, this graph is depicted in Figure~\ref{fig:tripartite_Y}.
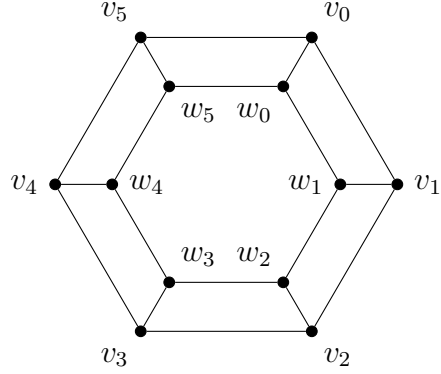
\begin{figure}[htb]
	\centering
	\begin{tikzpicture}[x=1.5cm, y=1.5cm, label distance=0cm]
		\vertex[fill] (w0) at (0.5,0.866) [label=265:$w_0$]{};
		\vertex[fill] (w1) at (1,0) [label=180:$w_1$]{};
		\vertex[fill] (w2) at (0.5,-0.866) [label=95:$w_2$]{};
		\vertex[fill] (w3) at (-0.5,-0.866) [label=85:$w_3$]{};
		\vertex[fill] (w4) at (-1,0) [label=0:$w_4$]{};
		\vertex[fill] (w5) at (-0.5,0.866) [label=275:$w_5$]{};
		
		\vertex[fill] (v0) at (0.75,1.299) [label=80:$v_0$]{};
		\vertex[fill] (v1) at (1.5,0) [label=0:$v_1$]{};
		\vertex[fill] (v2) at (0.75,-1.299) [label=280:$v_2$]{};
		\vertex[fill] (v3) at (-0.75,-1.299) [label=260:$v_3$]{};
		\vertex[fill] (v4) at (-1.5,0) [label=180:$v_4$]{};
		\vertex[fill] (v5) at (-0.75,1.299) [label=100:$v_5$]{};
			
		\path 
			(v0) edge (v1)
			(v1) edge (v2)
			(v2) edge (v3)
			(v3) edge (v4)
			(v4) edge (v5)
			(v5) edge (v0)

			(w0) edge (w1)
			(w1) edge (w2)
			(w2) edge (w3)
			(w3) edge (w4)
			(w4) edge (w5)
			(w5) edge (w0)
			
			(v0) edge (w0)
			(v1) edge (w1)
			(v2) edge (w2)
			(v3) edge (w3)
			(v4) edge (w4)
			(v5) edge (w5)
		;
	\end{tikzpicture}
	\caption{The graph $Y_{12}$ used in the construction of a minimal UPB in $\bb{C}^2 \otimes \bb{C}^2 \otimes \bb{C}^9$.}\label{fig:tripartite_Y}
\end{figure}

We now construct a set of vectors in $\bb{C}^2 \otimes \bb{C}^2$ whose orthogonality graph is $Y_{4k+4}$. To this end, we follow the notation of \cite{Fen06} and let $\{\ket{b_j},\ket{b_j^\perp}\}_{j=0}^{2k+1}$ be distinct orthonormal bases of $\bb{C}^2$ (i.e., $\braket{b_j}{b_j^\perp} = 0$ for all $j$, but $|\braket{b_i}{b_j}|,|\braket{b_i}{b_j^\perp}|,|\braket{b_i^\perp}{b_j^\perp}| \notin \{0,1\}$ whenever $i \neq j$). We then define the vectors $\ket{v_j^{(1,2)}},\ket{w_j^{(1,2)}} \in \bb{C}^2 \otimes \bb{C}^2$ as follows:
\begin{align*}
	\ket{v_{2j}^{(1,2)}} & := \ket{b_j} \otimes \ket{b_{2j}^\perp} & \ket{v_{2j+1}^{(1,2)}} & := \ket{b_j^\perp} \otimes \ket{b_{(2j+2) \, (\text{mod} (2k+2))}} \\
	\ket{w_{2j}^{(1,2)}} & := \ket{b_j^\perp} \otimes \ket{b_{2j+1}^\perp} & \ket{w_{2j+1}^{(1,2)}} & := \ket{b_j} \otimes \ket{b_{(2j+3) \, (\text{mod} (2k+2))}},
\end{align*}
for all $0 \leq j \leq k$. It is straightforward to verify that the set of vectors
\begin{align*}
	S^{(1,2)} := \big\{\ket{v_0^{(1,2)}},\ldots,\ket{v_{2k+1}^{(1,2)}},\ket{w_0^{(1,2)}},\ldots,\ket{w_{2k+1}^{(1,2)}}\big\}
\end{align*}
has orthogonality graph $Y_{4k+4}$. Furthermore, any product vector $\ket{z} \in \bb{C}^2 \otimes \bb{C}^2$ can be orthogonal to at most $2$ elements of $S^{(1,2)}$ on the first subsystem, and at most $1$ element of $S^{(1,2)}$ on the second subsystem. Thus any nonzero product vector can be orthogonal to at most $3$ elements of $S^{(1,2)}$.

We now consider the complement of the graph $Y_{4k+4}$, which we denote $X_{4k+4}$. This is a graph on the same set of $4k+4$ vertices $V$, but $(v_i,v_j)$ and $(w_i,w_j)$ are edges in $X_{4k+4}$ if and only if $j - i \notin \{1,-1\} \, (\text{mod } (2k+1))$, and $(v_i,w_j)$ is an edge in $X_{4k+4}$ if and only if $i \neq j$. Our goal is to show that there exists a set
\begin{align*}
	S^{(3)} = \big\{ \ket{v_0^{(3)}}, \ldots, \ket{v_{2k+1}^{(3)}}, \ket{w_0^{(3)}}, \ldots, \ket{w_{2k+1}^{(3)}} \big\} \subset \bb{C}^{4k+1}
\end{align*}
of $4k+4$ vectors with orthogonality graph $X_{4k+4}$. To this end, we let $S^{(3)}$ be the columns of the matrix $W$ described by Lemma~\ref{lem:tripartite_matrix}. It is clear from conditions~$(i)$, $(ii)$ and $(iii)$ of the lemma that these vectors have the desired orthogonality graph. Furthermore, condition~$(iv)$ guarantees that any nonzero state $\ket{z} \in \bb{C}^{4k+1}$ can be orthogonal to at most $4k$ vectors in $S^{(3)}$.

Because $X_{4k + 4} \cup Y_{4k + 4} = K_{4k+4}$, the complete graph on $4k+4$ vertices, the set of vectors
\begin{align*}
	S := \big\{\ket{v_0^{(1,2)}} \otimes \ket{v_0^{(3)}},\ldots,\ket{v_{2k+1}^{(1,2)}} \otimes \ket{v_{2k+1}^{(3)}},\ket{w_0^{(1,2)}} \otimes \ket{w_{0}^{(3)}},\ldots,\ket{w_{2k+1}^{(1,2)}} \otimes \ket{w_{2k+1}^{(3)}}\big\}
\end{align*}
is a product basis of $\bb{C}^2 \otimes \bb{C}^2 \otimes \bb{C}^{4k+1}$. To see that $S$ is unextendible, simply recall that any product state can be orthogonal to at most $3$ elements of $S$ on the first two subsystems and at most $4k$ elements of $S$ on the third subsystem, for a total of at most $4k + 3$ states total. Thus no nonzero product state can be orthogonal to all $4k+4$ elements of $S$, so unextendibility follows and the proof is complete.

We finish this section with the statement and proof of Lemma~\ref{lem:tripartite_matrix}, which played a key role in the above proof of Theorem~\ref{thm:tripartite_filler}.

\begin{lemma}\label{lem:tripartite_matrix}
	Let $k \geq 1$. There exist matrices $W_{(v)},W_{(w)} \in M_{4k+1,2k+2}$ such that if we define the block matrix $W := [W_{(v)}, W_{(w)}] \in M_{4k+1,4k+4}$ then:
	\begin{enumerate}[(i)]
		\item $W_{(v)}\ket{i}$ is orthogonal to $W_{(v)}\ket{j}$ whenever $j - i \notin \{0,1,2k+1\} \, (\text{mod } (2k+2))$;
		\item $W_{(w)}\ket{i}$ is orthogonal to $W_{(w)}\ket{j}$ whenever $j - i \notin \{0,1,2k+1\} \, (\text{mod } (2k+2))$;
		\item $W_{(v)}\ket{i}$ is orthogonal to $W_{(w)}\ket{j}$ whenever $j \neq i$; and
		\item every $(4k+1)\times (4k+1)$ submatrix of $W$ is nonsingular.
	\end{enumerate}
\end{lemma}
\begin{proof}
Define $d := 4k + 1$ for simplicity. Similar to the proof of Lemma~\ref{lem:cp_orthogonality}, to show the existence of $W$, it suffices to prove the following claim:

For any $1 \leq j_1 < \cdots < j_d \leq d+3$, there exists some $W^\prime \in M_{d,d+3}$ that satisfies the orthogonality conditions $(i)$, $(ii)$ and $(iii)$, and the $d \times d$ submatrix formed by columns $j_1,j_2,\ldots,j_{d}$ is nonsingular.

We first choose the $j_1,j_2,\ldots,j_{d}$-th column vectors to be the set of computational basis states in $d$-dimensional space. The submatrix formed by these columns is clearly nonsingular. We then fill other column vectors with suitable computational basis states one by one. To fill the $k$-th column vector, note that conditions $(i)$, $(ii)$ and $(iii)$ force it to be orthogonal to $(d-1)$ other columns of $W^\prime$. If we assume that $r$ of these $(d-1)$ columns have already been specified, then we can choose the $k$-th column vector to be any computational basis state that is orthogonal to the $r \leq d-1$ already-specified columns. By repeating this procedure, we eventually fill the entire matrix $W^\prime$ so that it satisfies conditions $(i)$, $(ii)$, and $(iii)$, and satisfies condition $(iv)$ for the particular submatrix formed by columns $j_1,j_2,\ldots,j_{d}$.
\end{proof}

\section{Explicit Construction of UPBs}\label{sec:explicit}

As the proofs of Lemmas~\ref{lem:hollow_unitary}, \ref{lem:cp_orthogonality}, and \ref{lem:tripartite_matrix} are non-constructive, it is perhaps not immediately clear how to produce explicit UPBs of the size indicated by Theorems~\ref{thm:main} and~\ref{thm:tripartite_filler}, even though we know they exist. We now address this problem and demonstrate how to construct explicit UPBs of the desired size in small dimensions. Code that implements the techniques described in this section, and thus constructs minimal UPBs, can be downloaded from \cite{JohUPBCode}.

For Lemma~\ref{lem:cp_orthogonality}, we recall that a generic matrix satisfying condition~$(i)$ of the lemma will also satisfy condition~$(ii)$. Thus, one way to construct matrices satisfying both requirements is to randomly generate its first $q$ columns, then generate its last $q$ columns according to the orthogonality requirement $(i)$, and finally check to make sure that the resulting matrix satisfies condition $(ii)$. All of these steps are straightforward to implement numerically in software such as MATLAB. For example, the following matrix $W_{5,3,2}$ satisfies the conditions of the lemma in the $q = 5$, $r = 3$, $s = 2$ case:
\begin{align*}
	W_{5,3,2} = \begin{bmatrix}[r]
     3 &  3 &  3 &  1 &  1 & -1 & -1 &  2 & -2 &  0 \\
     2 &  1 &  1 &  2 &  3 & -2 &  0 &  0 & -2 & -1 \\
     2 &  1 &  1 &  1 &  1 &  5 & -3 & -4 &  2 &  1 \\
     2 &  2 &  3 &  2 &  2 &  0 &  2 &  1 &  3 &  0
	\end{bmatrix}
\end{align*}

The procedure for Lemma~\ref{lem:tripartite_matrix} is similar -- just randomly generate columns that satisfy conditions~$(i)$, $(ii)$ and $(iii)$, and then with probability 1 the resulting matrix will also satisfy condition~$(iv)$. More specifically, randomly generate the first two columns of the matrix, then generate the remaining columns according to the orthogonality conditions, and finally check to make sure that condition $(iv)$ is satisfied. As before, this procedure is simple to perform numerically. The following matrix $W_2$ is an explicit example in the $k = 2$ (i.e., $d = 9$) case:
\begin{align*}
	W_{2} = \begin{bmatrix}[r]
		 2 &   1 &   1 &   5 &   1 &   2 &   0 &   0 &   1 &   0 &   6 &   6 \\
		 2 &  -1 &  -1 &  -3 &   1 &   2 &   0 &   0 &  -1 &   0 &   6 &   6 \\
		 1 &   2 &   0 &  -4 &   4 &   1 &   0 &   0 &   2 &   1 &   3 &   3 \\
		 2 &   2 &   0 &   0 &  -4 &  -1 &   0 &   0 &   2 &   1 &  27 &   6 \\
		 2 &   2 &   0 &   0 &   0 &   0 &  -1 &  -1 &  -4 &   0 &   0 &  -4 \\
		 1 &   2 &   0 &   0 &   0 &   0 &   1 &   2 &  -4 &   0 &   0 &  -2 \\
		 2 &   2 &   0 &   0 &   0 &   0 &   0 &   0 &   3 &  -1 & -20 & -14 \\
		 1 &   2 &   0 &   0 &   0 &   0 &   0 &   0 &   0 &  -1 &  21 &  23 \\
		 2 &   2 &   0 &   0 &   0 &   0 &   0 &   0 &   0 &   0 & -31 & -12
	\end{bmatrix}
\end{align*}

Unfortunately, constructing matrices that satisfy the constraints of Lemma~\ref{lem:hollow_unitary} seems to be a bit more difficult in practice. We can begin by specifying $(d-b)(d-b-1)/2$ of the non-zero entries and then using the orthogonality condition~$(i)$ to solve for the remaining non-zero entries. In general, however, this leads to a system of $(d-b)(d-b-1)/2$ linear and quadratic equations in $(d-b)(d-b-1)/2$ variables, which quickly becomes infeasible to solve as $d$ grows. Furthermore, the solution will generally be significantly messier than we saw when dealing with Lemmas~\ref{lem:cp_orthogonality} and~\ref{lem:tripartite_matrix}. Nonetheless, such a system of equations can be solved (even by hand) when $d$ is small enough. For example, in the $d = 6$, $b = 1$ case, if we set
\begin{align*}
	U_{6,1} = \begin{bmatrix}
		0 & u_1 & 1 & u_2 & 0 \\
		0 & 0 & u_3 & 1 & u_4 \\
		u_5 & 0 & 0 & u_6 & 1 \\
		1 & u_7 & 0 & 0 & 1 \\
		u_8 & 1 & u_9 & 0 & 0 \\
		1 & 1 & 1 & 1 & u_{10} \\
	\end{bmatrix}
\end{align*}
and then solve for the variables $\{u_j\}_{j=1}^{10}$, we find that the unique real solution is obtained when $u_{10}$ is the unique real root (approximately equal to $1.6445$) of the polynomial $3u_{10}^3 - 2u_{10}^2 - 3u_{10} - 3$, and the remaining $u_j$'s are given by
\begin{align*}
	u_1 & = (u_{10}-2)/(1-u_{10}) & u_4 & = u_{10}(u_{10}-2)/(2u_{10}-3) & u_7 & = -u_{10} \\
	u_2 & = (u_{10}-1)/(u_{10}-2) & u_5 & = -1 - u_{10} & u_8 & = u_{10}-1 \\
	u_3 & = (3-2u_{10})/(u_{10}-2) & u_6 & = 1/(1 + u_{10}) & u_9 & = 1/(1-u_{10}).
\end{align*}
It is a simple calculation in MATLAB to verify that $U_{6,1}$ also satisfies conditions~$(iii)$ and~$(iv)$ of Lemma~\ref{lem:hollow_unitary}, and thus leads to an unextendible product basis.

In the $b = 0$ case, however, we can simplify things a bit further by analytically constructing a matrix $U_{d,0}$ that satisfies conditions~$(i)$ and $(ii)$ of the lemma. We start by building the eigendecomposition of $U_{d,0}$ and then we argue that it must satisfy the condition~$(ii)$.

Let the first two eigenvalues of $U_{d,0}$ be $\lambda_1 := 1$ and $\lambda_2 := -1$, with respective eigenvectors
	\begin{align*}
		\ket{v_1} := \left[\tfrac{1}{\sqrt{2}}, \tfrac{1}{\sqrt{2d-2}}, \ldots, \tfrac{1}{\sqrt{2d-2}} \right]^T \text{ and } \ket{v_2} := \left[\tfrac{1}{\sqrt{2}}, \tfrac{-1}{\sqrt{2d-2}}, \ldots, \tfrac{-1}{\sqrt{2d-2}} \right]^T.
	\end{align*}
	It is easily-verified that $\ket{v_1}$ and $\ket{v_2}$ are orthogonal, and a simple calculation reveals that
	\begin{align*}
		\lambda_1\ketbra{v_1}{v_1} + \lambda_2\ketbra{v_2}{v_2} = \begin{bmatrix}0 & \tfrac{1}{\sqrt{d-1}} & \tfrac{1}{\sqrt{d-1}} & \cdots & \tfrac{1}{\sqrt{d-1}} \\ \tfrac{1}{\sqrt{d-1}} & 0 & 0 & \cdots & 0 \\ \tfrac{1}{\sqrt{d-1}} & 0 & 0 & \cdots & 0 \\ \vdots & \vdots & \vdots & \ddots & \vdots \\ \tfrac{1}{\sqrt{d-1}} & 0 & 0 & \cdots & 0 \end{bmatrix}.
	\end{align*}
	
	Now let $\omega_k := e^{2\pi i/k}$ be a primitive $k$-th root of unity. Let the remaining $d-2$ eigenvalues of $U_{d,0}$ be the $(d-2)$-th roots of unity: $\lambda_{j+2} := \omega_{d-2}^j$ for $j = 1, 2, \ldots, d-2$. Define the corresponding eigenvectors as follows:
	\begin{align*}
		\ket{v_{j+2}} := \tfrac{1}{\sqrt{d-1}}\left[0, \omega_{d-1}^0, \omega_{d-1}^j, \omega_{d-1}^{2j}, \ldots, \omega_{d-1}^{j(d-2)} \right]^T.
	\end{align*}
	In other words, $\big[\ket{v_{3}},\ket{v_{4}},\ldots,\ket{v_{d}}\big]$ is the $(d-1)\times(d-1)$ Fourier matrix, with its leftmost column removed and a row of zeroes added to the top. The fact that $\{\ket{v_j}\}_{j=1}^d$ forms an orthonormal basis of $\bb{C}^d$ comes from the fact that the $(d-1)\times(d-1)$ Fourier matrix is unitary. It follows that $U_{d,0} := \sum_{j=1}^d \lambda_j \ketbra{v_j}{v_j}$ is unitary as well. Furthermore, a straightforward (albeit tedious) calculation reveals that if $k,\ell \geq 1$ then
	\begin{align*}
		\bra{k}U_{d,0}\ket{\ell} & = \frac{1}{d-1}\sum_{j=1}^{d-2} \lambda_{j+2} \omega_{d-1}^{(k-2)j}\omega_{d-1}^{-(\ell-2)j} \\
		& = \frac{1}{d-1}\sum_{j=1}^{d-2} \omega_{d-2}^j \omega_{d-1}^{(k-\ell)j} \\
		& = \frac{1}{d-1}\sum_{j=1}^{d-2} \omega_{(d-1)(d-2)}^{j(d-1+(k-\ell)(d-2))} \\
		& = \frac{\omega_{(d-1)(d-2)}^{(d-2)(d-1+(k-\ell)(d-2))}-1}{(d-1)(1-\omega_{(d-1)(d-2)}^{1-d-(k-\ell)(d-2)})},
	\end{align*}
	where we summed the geometric series in the final line above. This quantity equals zero if and only if $(d-2)(d-1+(k-\ell)(d-2))$ is a multiple of $(d-1)(d-2)$ (i.e., if and only if $(k-\ell)(d-2)$ is a multiple of $d-1$). However, because $0 \leq |k-\ell| < d-1$ and $d-2$ is coprime to $d-1$, this happens if and only if $k = \ell$. This verifies that $U_{d,0}$ indeed has zeroes down its diagonal, as desired, and also shows that it is a candidate to satisfy condition~$(iii)$ of the lemma (since in the $b = 0$ case, any matrix satisfying condition~$(iii)$ must have all of its off-diagonal entries be non-zero). We have numerically verified that $U_{d,0}$ satisfies condition~$(iii)$ (and it trivially satisfies condition~$(iv)$) for $4 \leq d \leq 19$ in MATLAB.
	
	By using all of these techniques, together with the techniques used in the proof of Theorem~\ref{thm:main}, we can explicitly construct UPBs in all of the dimensions we have discussed. For example, a minimal UPB of size $10$ in $\bb{C}^4 \otimes \bb{C}^6$ is given by the set
	\begin{align*}
		\big\{ \ket{v_0^{(1)}} \otimes \ket{v_0^{(2)}}, \ldots, \ket{v_4^{(1)}} \otimes \ket{v_4^{(2)}}, \ket{w_0^{(1)}} \otimes \ket{w_0^{(2)}}, \ldots, \ket{w_4^{(1)}} \otimes \ket{w_4^{(2)}}\big\},
	\end{align*}
	where:
	\begin{itemize}
		\item the $\ket{v_j^{(1)}}$'s and $\ket{w_j^{(1)}}$'s are (in order) the normalized columns of $W_{5,3,2}$ above;
		\item $\ket{v_j^{(2)}} = \ket{j} \in \bb{C}^6$ for $0 \leq j \leq 4$; and
		\item $\ket{w_j^{(2)}}$ is the $(j+1)$-th column of $U_{6,1}$ above.
	\end{itemize}

\section{Outlook}\label{sec:conclusions}

We have shown that, in many cases, the minimum size of a UPB does not exceed the trivial lower bound by more than $1$. In fact, there is currently no known case in which $f_m(d_1,\ldots,d_p) > f_N(d_1,\ldots,d_p) + 1$. It could be the case that this never happens, or it could be the case that we aren't aware of any such cases yet because it is very difficult to prove non-trivial lower bounds on $f_m(d_1,\ldots,d_p)$.

Some particularly interesting cases of the minimal UPB question that remain open are:
\begin{enumerate}[(1)]
	\item $d_1 = d_2 = 2$, $d_3 = 4k - 1$: It was shown in \cite{Fen06} that $f_m(2,2,4k-1) = f_N(2,2,4k-1) + 1$ when $k = 1$, but the proof technique does not seem to generalize straightforwardly to the $k \geq 2$ case.
	
	\item $p = 4k$ and $d_1 = \cdots = d_p = 2$: All other cases with $d_1 = \cdots = d_p = 2$ have been solved -- it is known that $f_m(2,\ldots,2) = f_N(2,\ldots,2)$ if $p$ is odd and $f_m(2,\ldots,2) = f_N(2,\ldots,2) + 1$ if $p \equiv 2 \, (\text{mod } 4)$. Furthermore, it is known that $f_m(2,2,2,2) = f_N(2,2,2,2) + 1$, but again the proof technique does not obviously generalize to the $k \geq 2$ case.
	
	\item $d_1 = 3$, $d_2 = d_3 = 4$: Excluding the open case (1) above, this is now the smallest unsolved tripartite case.
\end{enumerate}

Finally, it was noted in \cite{AL01} that whenever $f_m(d_1,\ldots,d_p) = f_N(d_1,\ldots,d_p)$, a minimal UPB can be constructed using only real vectors. The same is true in all cases in which we have proved $f_m(d_1,\ldots,d_p) = f_N(d_1,\ldots,d_p) + 1$. Restricting to the real field cuts the number of variables in half and also cuts the number of polynomial constraints in half, but has no significant effect on any of our proofs.

\vspace{0.1in} \noindent{\bf Acknowledgements.} Thanks are extended to John Watrous for a helpful discussion. J.C. was supported by NSERC, UTS-AMSS Joint Research Laboratory for Quantum Computation and Quantum Information Processing and NSF of China (Grant No. 61179030). N.J. was supported by an NSERC Postdoctoral Fellowship. 

\bibliographystyle{alpha}
\bibliography{quantum}
\end{document}